\def\arxiv{}
\pdfoutput=1
\documentclass[sigconf,authorversion]{aamas}

\usepackage{balance}
\usepackage{amsfonts}
\usepackage{xfrac}
\usepackage{booktabs}
\usepackage{tikz}
\usepackage{subcaption}
\usepackage[ruled, vlined, linesnumbered]{algorithm2e}
\usepackage[noabbrev, nameinlink]{cleveref}

\makeatletter
\gdef\@copyrightpermission{
  \begin{minipage}{0.2\columnwidth}
   \href{https://creativecommons.org/licenses/by/4.0/}{\includegraphics[width=0.90\textwidth]{by}}
  \end{minipage}\hfill
  \begin{minipage}{0.8\columnwidth}
   \href{https://creativecommons.org/licenses/by/4.0/}{This work is licensed under a Creative Commons Attribution International 4.0 License.}
  \end{minipage}
  \vspace{5pt}
}
\makeatother

\ifdefined\arxiv
\setcopyright{none}
\acmConference{Proc.\@ of the 24th International Conference
on Autonomous Agents and Multiagent Systems (AAMAS 2025), \url{https://dl.acm.org/doi/10.5555/3709347.3743752}}{May 19 -- 23, 2025}
{Detroit, Michigan, USA}{Y.~Vorobeychik, S.~Das, A.~Nowé  (eds.)}
\else
\setcopyright{ifaamas}
\acmConference[AAMAS '25]{Proc.\@ of the 24th International Conference
on Autonomous Agents and Multiagent Systems (AAMAS 2025)}{May 19 -- 23, 2025}
{Detroit, Michigan, USA}{Y.~Vorobeychik, S.~Das, A.~Nowé  (eds.)}
\fi
\copyrightyear{2025}
\acmYear{2025}
\acmDOI{}
\acmPrice{}
\acmISBN{}

\acmSubmissionID{470}

\newcommand{\s}{\mathbf{s}}
\newcommand{\x}{\mathbf{x}}
\renewcommand{\t}{\mathbf{t}}

\newcommand{\A}{\mathcal{A}}
\newcommand{\B}{\mathcal{B}}
\newcommand{\M}{\mathcal{M}}
\renewcommand{\L}{\mathcal{L}}
\definecolor{cred}{HTML}{D81B60}
\definecolor{cblue}{HTML}{1E88E5}
\definecolor{cyellow}{HTML}{D09C00}
\definecolor{cgreen}{HTML}{5B8600}
\definecolor{cgray}{HTML}{AAAAAA}

\newcommand*{\trlength}{1}
\newcommand*{\locsepdist}{0.5*\trlength}
\newcommand*{\noderadius}{7.5pt}
\usetikzlibrary{arrows.meta, quotes, calc, fit, decorations.pathreplacing}
\tikzset{
node distance={\trlength cm}, vert/.style = {draw, circle, inner sep = 0 pt, minimum size = 2 * \noderadius, text depth=0.0ex, text height=1.25ex}, every path/.style = {-Latex}, every label/.append style={rectangle, font = {\footnotesize}}, miller/.style = {circle,fill,inner sep=2.5pt, cred}, baker/.style = {rectangle,fill,inner sep=2.5pt, cblue}, every node/.style = {font = {\footnotesize}}, every edge quotes/.style = {auto, sloped, font = {\footnotesize}}, unchosen/.style = {-, dashed}
}

\title{The Bakers and Millers Game with Restricted Locations}

\author{\href{https://orcid.org/0000-0001-6577-6756}{Simon Krogmann}}
\affiliation{
  \institution{Hasso Plattner Institute,\\University of Potsdam}
  \city{Potsdam}
  \country{Germany}}
\email{simon.krogmann@hpi.de}

\author{\href{https://orcid.org/0000-0002-3010-1019}{Pascal Lenzner}}
\affiliation{
  \institution{Institute of Computer Science, University of Augsburg}
  \city{Augsburg}
  \country{Germany}}
\email{pascal.lenzner@uni-a.de}

\author{\href{https://orcid.org/0000-0002-4950-8708}{Alexander Skopalik}}
\affiliation{
  \institution{Mathematics of Operations Research, University of Twente}
  \city{Enschede}
  \country{The Netherlands}}
\email{a.skopalik@utwente.nl}

\begin{abstract}
We study strategic location choice by customers and sellers, termed the Bakers and Millers Game in the literature. In our generalized setting, each miller can freely choose any location for setting up a mill, while each baker is restricted in the choice of location for setting up a bakery.
For optimal bargaining power, a baker would like to select a location with many millers to buy flour from and with little competition from other bakers. Likewise, a miller aims for a location with many bakers and few competing millers. Thus, both types of agents choose locations to optimize the ratio of agents of opposite type divided by agents of the same type at their chosen location. Originally raised in the context of Fractional Hedonic Games, the Bakers and Millers Game has applications that range from commerce to product design.

We study the impact of location restrictions on the properties of the game.
While pure Nash equilibria trivially exist in the setting without location restrictions, we show via a sophisticated, efficient algorithm that even the more challenging restricted setting admits equilibria. Moreover, the computed equilibrium approximates the optimal social welfare by a factor of at most $2\left(\frac{e}{e-1}\right)$. Furthermore, we give tight bounds on the price of anarchy/stability.

On the conceptual side, the location choice feature adds a new layer to Hedonic Games, in the sense that agents that select the same location form a coalition. This allows to naturally restrict the possible coalitions that can be formed. With this, our model generalizes simple symmetric Fractional Hedonic Games on complete bipartite valuation graphs and also Hedonic Diversity Games with utilities single-peaked at 0. We believe that this generalization is also a very interesting direction for other types of Hedonic Games.
\end{abstract}

\begin{document}
\pagestyle{fancy}
\fancyhead{}

\maketitle


\section{Introduction}
Markets facilitate trade. They are the key places where supply meets demand. Typically, there are many possible markets for goods or services, and thus sellers face the economic decision of which markets to serve, while customers strategically decide which market to patronize. The choice of the right location for trade is crucial in a competitive economic environment both for sellers and for customers. The number of competitors in a given market is a key factor in determining supply and demand, and eventually, prices and profits from trade.
For achieving high prices, sellers want to minimize competition with other sellers and at the same time, they want to maximize the number of customers. For customers who aim for low prices, it is important to patronize a market with many sellers to choose from and with low competition from other customers. 

\begin{figure}[b]
    \centering
    \begin{tikzpicture}
        \node at (4, 0) {locations};
        \node at (4, -\trlength) {millers};
        \node at (4, \trlength) {bakers};
        
        \draw[dotted,-] (-0.2,\locsepdist) -- (5+3.2,\locsepdist);
        \draw[dotted,-] (-0.2,-\locsepdist) -- (5+3.2,-\locsepdist);
        
        \foreach \offset in {0,5} {
            \node[vert] (x\offset) at (\offset + 0.5, 0) {$x$};
            \node[vert] (y\offset) at (\offset + 1.5, 0) {$y$};
            \node[vert] (z\offset) at (\offset + 2.5, 0) {$z$};
            
            \node[miller] [label=left:$m$] (m0\offset) at (\offset + 1, -\trlength) {};
            \node[miller] (m1\offset) at (\offset + 2, -\trlength) {};

            \foreach \i in {0,1,2,3} {
                \node[baker] (b\i\offset) at (\offset + 1*\i, \trlength) {};
            }
        }
        \path[unchosen] (b00) edge (x0);
        \path[unchosen] (b10) edge (x0);
        \path[unchosen] (b20) edge (z0);
        \path (b00) edge (y0);
        \path (b10) edge (y0);
        \path (b20) edge (y0);
        \path (b30) edge (z0);

        \path (m00) edge (x0);
        \path (m10) edge (y0);
        
        \path[unchosen] (b05) edge (x5);
        \path[unchosen] (b15) edge (x5);
        \path[unchosen] (b25) edge (y5);
        \path (b05) edge (y5);
        \path (b15) edge (y5);
        \path (b25) edge (z5);
        \path (b35) edge (z5);

        \path (m05) edge (y5);
        \path (m15) edge (z5);
    \end{tikzpicture}
    \caption{Illustration of the Bakers and Millers Game with restricted locations. Bakers choose between an agent-specific set of locations (black arrows mark the selected location, dashed lines the unselected ones), while millers can choose any location.
    On the left, miller $m$ can improve her utility by choosing location~$y$ for utility $\frac32$ or location~$z$ for utility $1$.
    A pure Nash equilibrium is given on the right.}
    \label{fig:example-bakers}
\end{figure}
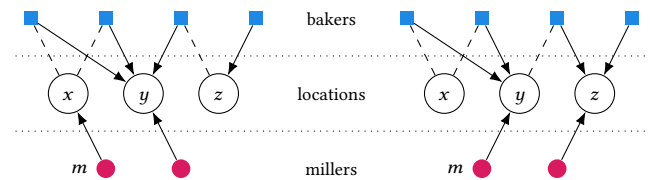

Our running example of the above competitive setting will be the Bakers and Millers Game, as briefly introduced by \citet*{bakermiller_conf}. There, millers choose locations to set up their mill to produce flour whereas bakers depend on buying their flour from local millers. In this scenario, millers aim for locations with many bakers as customers and only a few other competing millers. Bakers will patronize a local miller that has as few competing customers as possible. If both bakers and millers can freely choose any location, then this setting is well-studied, since it corresponds to a simple symmetric Fractional Hedonic Game on a complete bipartite valuation graph~\cite{bakermiller_conf,bakermiller}. However, in many practical applications, there is an asymmetry between sellers and customers, i.e., sellers can freely choose where to set up their stores, but customers can only patronize a restricted set of locations, e.g., their local neighborhood.
In this paper, we study the impact of such an asymmetry of location choice by investigating the Bakers and Millers Game where one side has restricted locations. In our setting, the millers can freely choose their location, while the bakers can only choose from an agent-specific restricted set of locations. See \Cref{fig:example-bakers} for an illustration. With this, we investigate a very simple model that naturally extends simple symmetric Fractional Hedonic Games on complete bipartite valuation graphs.

More examples are classical farmer's markets but also online commerce on different platforms like Alibaba, Amazon, eBay or Zalando. Another application is gaining a favorable position in a two-sided matching market:
A worker wants to move to a city that has more companies and thus more possible jobs to choose from, whereas a company wants to move to a city which has fewer companies that compete for the local workforce. However, location choice is not restricted to trade. "Sellers" could be product designers, politicians, researchers, or writers who strategically focus on certain topics, and the "customers" might be consumers, voters, reviewers or readers only interested in some of these topics. Thus, "locations" might simply be points in some abstract feature space.

Given the importance of strategic location choice, the research on such problems dates back roughly a century, starting with Hotelling's seminal model of several sellers competitively selecting a location for their store on the city's main street~\cite{hotelling} that was later adopted by Downs for modeling political competition~\cite{downs}. Based on this, many other location choice models have been investigated in Economics, Mathematics, Computer Science and Operations Research and this field is typically called Location Analysis~\cite{RE05,ELT93}. However, most of these models solely focus on the seller's side, i.e., on locating stores in settings where the customers do not face strategic decisions since they simply patronize their nearest store.
Also, the opposite point of view has been taken, where the sellers have fixed positions and the focus is on the customer behavior, like in consumer behavior theory~\cite{blundell1988consumer,schiffman2013consumer}.

The setting of rational agents who compete over a fixed set of offered services, products or resources is covered by the broad area of Congestion Games~\cite{Rosenthal1973,MS96}. There, agents strategically select resources and an agent's utility for a given resource deteriorates with the number of other agents who also choose the same resource. Thus, in the Bakers and Millers Game, both bakers and millers face a Singleton Congestion Game if the behavior of the other type of agents is fixed, i.e., given the location choice of the millers, the bakers want to select locations with a minimum number of other bakers and a maximum number of millers and vice versa. Hence, our game combines both the sellers' and the customers' perspectives simultaneously, i.e., both sides strategically select a location which then determines which customers are matched to which sellers. Our game focuses on the impact of the location choice and it abstracts away from the exact mechanisms that determine prices via supply and demand. 
This allows for modeling a broad range of applications ranging from commerce, over two-sided matching, to the design of products or political campaigns.

\subsection{Our Contribution}
We study competitive location choice in the Bakers and Millers Game introduced by \citet*{bakermiller_conf}. Most importantly, we generalize this setting by considering that one side, the bakers, may have location restrictions, i.e., they can only choose their location from an individually predetermined subset of possible locations. This naturally models an asymmetry in many competitive settings where sellers or producers can freely choose where to open stores or facilities while the customers are bound to patronize stores accessible to them only, e.g., local stores in their vicinity. Besides store locations, this can also model product or feature choices by producers in a setting where the customers are only interested in certain subsets of features or products, like family-friendly cars, lightweight e-bikes, or regional vegan food.

The Bakers and Millers Game with restricted locations generalizes simple symmetric Fractional Hedonic Games on complete bipartite valuation graphs by naturally limiting the possible coalitions that can be formed. The same holds for Hedonic Diversity Games with utilities that have a single peak at 0. As these two examples indicate, we emphasize that this extra layer of first selecting locations that then determine the formed coalitions can be added to any variant of Hedonic Games to better model real-world settings.

As our main contribution, we analyze the impact of the location restrictions in the Bakers and Millers Game. While proving the existence of pure Nash equilibria is trivial if both bakers and millers can freely choose their locations, this task is much more challenging if the bakers have restricted locations. Despite this, we show the existence of pure Nash equilibria by giving a polynomial-time algorithm to compute one. As a consequence of the location restrictions, some bakers might end up without an available miller. Hence, we use the \emph{coverage of bakers}, i.e., the number of bakers with at least one available miller, as our social welfare function. For this function, our aforementioned efficient algorithm computes a pure Nash equilibrium that approximates the optimal social welfare by $\left(1 + \frac{\min(|\L|,|\M|)-1}{|\M|}\right)\frac{e}{e-1} < 2\left(\frac{e}{e-1}\right)$ with Euler's number $e$ and the sets of locations $|\L|$ and millers $|\M|$.
On the other hand, we show that computing a strategy profile or even a pure Nash equilibrium with optimal social welfare is NP-hard.
Finally, we establish tight bounds of $1 + \frac{\min(|\L|,|\M|)-1}{|\M|} < 2$ on the price of stability and of $|\B|$ on the price of anarchy, with the set of bakers $\B$.

\subsection{Preliminaries and Model}

\paragraph{\bf Basics.}
We denote the set of available locations by $\L$ and the set of agents by $\A \coloneq \B \cup \M$, with the set of bakers $\B$ and millers $\M$.
Each baker $b\in \B$ has a set of feasible locations $L(b) \subseteq \L$.

A strategy profile is given by the tuple $(\s, \t)$, where $\s$ is a vector of bakers' locations $(s_b)_{b \in \B}$, with $s_b \in L(b)$, and $\t$ is a vector of millers' locations $(t_m)_{m \in \M}$, with $t_m \in L$.
We denote by $B_\s(\ell) \coloneq |\{b \in \B \mid s_b = \ell\}|$ the number of bakers at location $\ell$ for a given strategy profile $\s$.
Similarly, $M_\t(\ell) \coloneq |\{m \in \M \mid t_m = \ell\}|$ denotes the number of millers at location $\ell$.

\paragraph{\bf Objectives.}
Agents choose locations to avoid competition due to agents of the same type while maximizing the availability of agents of the other type.
In its simplest form, this notion is captured by a utility function that corresponds to the ratio of the two types of agents.
Hence, the (always well-defined) utility of a baker $b$ and a miller $m$ in a given strategy profile $(\s, \t)$ is
\[
    u_b(\s, \t) \coloneq \frac{M_\t(s_b)}{B_\s(s_b)}
\quad\text{and}\quad
    u_m(\s, \t) \coloneq \frac{B_\s(t_m)}{M_\t(t_m)} \text{, respectively.}
\]

\paragraph{\bf Equilibria.}
A state $(\s^*, \t^*)$ is a \emph{miller equilibrium} if for each miller $m \in \M$ there is no location $t_m'\in \L$ such that
\[
u_m(\s^*, (t_m', \t^*_{-m})) > u_m(\s^*, \t^*)\text.
\]
A state $(\s^*, \t^*)$ is a \emph{baker equilibrium} if for each baker $b \in \B$ there is no location $s_b' \in L(b)$ such that
\[
u_b( (s_b', \s^*_{-b}), \t^*) > u_b(\s^*, \t^*)\text.
\]
A state $(\s^*, \t^*)$ is a pure Nash equilibrium if it is both a baker equilibrium and a miller equilibrium.

\paragraph{\bf Efficiency.}
We measure the efficiency of equilibria by the number of bakers that can shop at some miller.
Therefore, we define the \emph{coverage} of a strategy profile $(\s,\t)$ as the number of bakers with at least one miller at the same location, i.e.,
\[
    W(\s, \t) \coloneq |\{b \in \B \mid M_\t(s_b) > 0\}|\text.
\]
In \Cref{sec:welfare}, we justify this definition.
Based on the coverage, we define the \emph{price of anarchy}~\cite{poa} and the \emph{price of stability}~\cite{pos} as
\[
    \text{PoA} \coloneq \sup_{I} \left(\frac{W(\text{OPT}(I))}{W(\text{worstNE}(I)}\right)
    \quad\text{and}\quad
    \text{PoS} \coloneq \sup_{I} \left(\frac{W(\text{OPT}(I))}{W(\text{bestNE}(I)}\right)\text,
\]
where for an instance $I$, we have $\text{OPT}(I)$ as the social optimum, $\text{worstNE}(I)$ as the Nash equilibrium with the lowest coverage and $\text{bestNE}(I)$ as the Nash equilibrium with the highest coverage.

\subsection{Related Work}
Research in the broad field of Hedonic Games was initiated by the seminal works of \citet*{DG80} and \citet*{BJ02}. In these games, the agents have preferences over all possible coalitions of agents and every partition of the set of agents into disjoint coalitions is considered an outcome. The crucial aspect of Hedonic Games is that the utility of an agent in some coalition only depends on the composition of her coalition. Given this, many variants of utility functions have been studied.

\paragraph{\bf Fractional Hedonic Games.}
Closest to our model are Fractional Hedonic Games (FHGs) introduced by~\citet*{bakermiller_conf} and later extended~\cite{bakermiller}, where each agent has an individual value for every other agent and the utility of some agent in a coalition is her average agent value within her coalition, i.e., it is the sum over the individual values of all the agents in her coalition divided by the number of agents in the coalition.

A nice feature of FHGs is that an instance can be described via a weighted directed \emph{valuation graph}, where an edge~$(u,v)$ with weight~$w$ encodes that agent~$u$ has a value of~$w$ for agent~$v$. A FHG is called \emph{symmetric} if an edge $(u,v)$ with weight $w$ in the corresponding valuation graph implies the existence of the reverse edge $(v,u)$ with the same weight $w$, i.e., pairs of agents value each other equally. A FHG is called \emph{simple} if all edge weights in the valuation graph are either $0$ or $1$. Thus, the valuations in simple symmetric FHGs can be modeled with an undirected unweighted valuation graph.

If the number of locations in our game equals the total number of bakers and millers and if the bakers can choose any location, then our game is equivalent to a simple symmetric FHG.
On the other hand, if the valuation graph of a simple symmetric FHG is bipartite and complete, then this is equivalent to our game.\footnote{FHGs and our game differ slightly in the definition of the agents' valuations. However, it is easy to show that the agents' preferences over coalitions/locations are equivalent.} Thus, our model generalizes simple symmetric FHGs that have bipartite and complete valuation graphs.

FHGs have been intensively studied with different notions of stability, like core, Nash or individual stability~\cite{BJ02,bakermiller,bilo2018nash,fhg-individual,carosi2019,DBLP:conf/wine/MonacoM23} and also in a modified version where the agents do not count themselves in the computation of their utility in a coalition~\cite{monaco2019,DBLP:journals/aamas/MonacoMV20}.
In our paper, we use Nash stability. In a Nash stable partition, no agent can improve her utility by unilaterally changing her coalition.

Nash stable partitions are not guaranteed to exist for FHGs with arbitrary agent valuations, but they do exist for non-negative valuations~\cite{bilo2018nash}, which include simple FHGs.
Also, it is NP-complete to decide if a given instance admits a Nash stable partition, even with symmetric valuations~\cite{fhg-individual}. If the quality of stable partitions is measured with the utilitarian welfare, then the price of anarchy is in $\Theta(n)$, even on unweighted paths, and the price of stability on weighted graphs is in $\Theta(n)$, even on weighted stars~\cite{bilo2018nash}. Note that we also prove bounds on the price of anarchy and stability but with a different social welfare function. \citet*{bilo2018nash} show that best response dynamics are not guaranteed to converge to a Nash stable partition even on unweighted bipartite graphs and that computing a Nash stable partition with maximum social welfare is NP-hard. However, the non-convergence result does not carry over to our game since the bipartite valuation graph in~\cite{bilo2018nash} is not complete.

\paragraph{\bf Hedonic Diversity Games.}
Also related are Hedonic Diversity Games (HDGs)~\cite{BredereckEI19,BoehmerE20}.  There, the agents have different types and their utility depends on the type-composition of their coalition. Similarly to FHGs, the utility of an agent is the fraction of own-type agents. The utilities in our game correspond to single-peaked utilities, a concept that dates back to \citet{black1948rationale}, and such preferences have also been investigated in HDGs~\cite{BredereckEI19,BoehmerE20}.

As with FHGs, our game is equivalent to HDGs with utilities that are single-peaked at $0$ if the number of locations equals the total number of bakers and millers and if all bakers can choose any location. Thus, by restricting the number of locations or the access to these locations, our game generalizes such HDGs.

\paragraph{\bf Schelling Games and Resource Selection.}
Another class of related games are Schelling Games~\cite{ChauhanLM18,Echzell0LMPSSS19,AgarwalEGISV21}. There, agents of different types select a location on a given graph and the utility of an agent is a function of the type composition of the neighborhood on her selected location. The crucial difference to FHGs and HDGs is that these neighborhoods may be non-overlapping. The agents' utility is a threshold function that attains its maximum value if the fraction of same-type agents in the neighborhood is at least $\tau$, for some $\tau \in [0,1]$. Recently, single-peaked utility functions have also been studied~\cite{BBLM22_single-peaked,FLMS23_single-peaked}.

Close to Schelling Games and to our game is the model presented by~\citet{ijcai-23}. In their setting, agents of different types jointly select resources. Like in Schelling Games, an agent is content with her resource if the fraction of same-type agents selecting the resource reaches at least a tolerance threshold $\tau \in [0,1]$. Two variants are studied: impact-blind and impact-aware agents. For the former, equilibrium existence is shown via a potential function argument while for the latter, this only holds for $\tau\leq\frac{1}{2}$.

By suitably modifying the utilities and assuming $\tau=1$, it can be shown that our Bakers and Millers Game is a special case of the model with impact-aware agents in~\cite{ijcai-23} with inverted utilities. However, since their potential argument only works for $\tau\leq \frac{1}{2}$, their existence result does not carry over to our game. Also, in the other direction, our positive results do not carry over to their model.

\paragraph{\bf Strategic Facility Location.}
Finally, strategic facility location models~\cite{OwenD98} are also related, in particular, the model by~\citet{Vetta02}, where sellers strategically select markets with a certain purchasing power and the price is explicitly determined by the delivery costs. Voronoi Games~\cite{DurrT07,FeldmannMM09}, Location Games on Networks~\cite{FournierS19}, Market Sharing Games~\cite{GoemansLMT06}, and Network Investment Games~\cite{Schmand0S19} are also similar since agents strategically select a location or resources to maximize their utility. Closer to our setting are models where the clients also face a strategic choice in the form of minimizing a weighted sum of their travel and waiting time~\cite{kohlberg1983equilibrium,PetersSV18,FeldottoLMS19}. Recently, this setting with both strategic facilities and strategic clients was studied as a sequential game~\cite{KrogmannLMS21,KrogmannLS23,Krogmann24}, where first sellers select a location on a given graph and then the customers on graph nodes decide how to distribute their purchasing power among neighboring sellers. Our model resembles a simultaneous variant of this game with different utility functions.

\section{Equilibrium Existence and Computation}

We observe that equilibria in the Bakers and Millers Game are not unique and can be quite different.
See \Cref{fig:multiple-ne} for an example with two locations $x$ and $y$, four bakers and two millers.
In the left equilibrium, all agents except one baker are at the same location~$x$.
On the right, each location has one miller and two bakers.

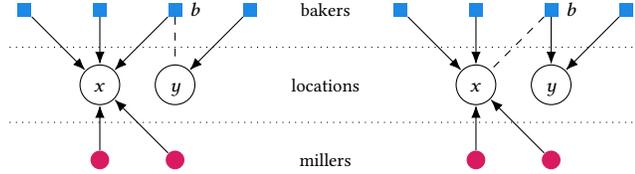
\begin{figure}[h]
    \centering
    \begin{tikzpicture}
        \node at (3., 0) {locations};
        \node at (3., -\trlength) {millers};
        \node at (3., \trlength) {bakers};
        
        \draw[dotted,-] (-1.2,\locsepdist) -- (5+2.2,\locsepdist);
        \draw[dotted,-] (-1.2,-\locsepdist) -- (5+2.2,-\locsepdist);

        \foreach \offset / \name in {0/l,5/r} {
            \node[vert] (x\name) at (\offset+0, 0) {$x$};
            \node[vert] (y\name) at (\offset+1, 0) {$y$};
            
            \node[miller] (f1\name) at (\offset+0, -\trlength) {};
            \node[miller] (f2\name) at (\offset+1, -\trlength) {};

            \node[baker] (c1\name) at (\offset-1, \trlength) {};
            \node[baker] (c2\name) at (\offset+0, \trlength) {};
            \node[baker, label=right:$b$] (c3\name) at (\offset+1, \trlength) {};
            \node[baker] (c4\name) at (\offset+2, \trlength) {};
        }

        \path (f1l) edge (xl);
        \path (f2l) edge (xl);
        
        \path (c1l) edge (xl);
        \path (c2l) edge (xl);
        \path (c3l) edge (xl);
        \path (c4l) edge (yl);

        \path[unchosen] (c3l) edge (yl);

        \path (f1r) edge (xr);
        \path (f2r) edge (yr);
        
        \path (c1r) edge (xr);
        \path (c2r) edge (xr);
        \path (c3r) edge (yr);
        \path (c4r) edge (yr);

        \path[unchosen] (c3r) edge (xr);
    \end{tikzpicture}
    \caption{An instance with two Nash equilibria which differ in the agents' utilities and the coverage. Of the bakers, only $b$ has both locations available in $L(b)$.}
    \label{fig:multiple-ne}
\end{figure}

In the special case of $L(b)=\L$, i.e., where each baker can choose every location, there is always a trivial pure Nash equilibrium by assigning all agents to a single location.
Thus, we now turn to the non-trivial case with restricted bakers and devise an algorithm to compute a pure Nash equilibrium.
This algorithm, shown in \Cref{alg:ne}, proceeds in three phases:
First, it tries to concentrate the bakers as much as possible in few locations so that they are placed in the fullest location they can access.
Then, a miller equilibrium is created by sequentially inserting each miller at the currently best location.
Finally, the bakers are redistributed to a specific baker equilibrium, which maintains the miller equilibrium.
We show correctness in the following proof with two embedded lemmas.

\makeatletter
\patchcmd{\@algocf@start}
  {-1.5em}
  {0pt}
  {}{}
\makeatother
\begin{algorithm}
    \caption{compute an equilibrium profile $(\s^*,\t^*)$}
    \label{alg:ne}
    $\L'\gets \L$, $\B' \gets \B$\;
    \For{$i=1$ to $|\L|$}{
        $\ell_i \gets \arg\max_{\ell \in \L'} |\{b \in \B' \mid \ell \in L(b) \}|$, a location with most bakers in range\;
        \For{baker $b \in \B'$ with $\ell_i \in L(b)$}{
            $s_b \gets \ell_i$\;
        }
        $\L' \gets \L' \setminus \{\ell_i\}$, $\B' \gets \B' \setminus \{b \mid \ell_i \in L(b)\}$\;
    }
    \For{miller $m \in \M$}{
        $t^*_m \gets$ location $\ell$ (with smallest index) maximizing $\frac{B_\s(\ell)}{M_\t(\ell) + 1}$\;\label{line:miller-insertion}
    }
    $\s^* \gets$ baker strategy profile maximizing $\Phi_{\t^*}(\s^*)$ (\Cref{lem:compute-maximizer})\;\label{line:maximizer}
\end{algorithm}

\begin{theorem}[Nash Equilibrium Algorithm]
    \label{thm:ne}
    \Cref{alg:ne} computes a pure Nash equilibrium in polynomial time.
\end{theorem}
\begin{proof}
\Cref{alg:ne} first greedily determines the order of the locations by the number of yet unassigned bakers in range and assigns the bakers to those locations in that order resulting in a baker profile $\s$.
See \Cref{fig:loop1} for an example.
We then iteratively insert millers to their best-response locations.
This yields a profile $(\s,\t^*)$ in which the millers are in equilibrium.
Finally, we compute a baker equilibrium profile $\s^*$ by choosing the state $s$ that maximizes the classic Rosenthal potential function~\cite{Rosenthal1973}
\[
    \Phi_{\t^*}(s) \coloneq \sum_{\ell \in \L} M_{\t^*}(\ell) \cdot H_{B_s(\ell)}\text,
\]
where $H_i$ denotes the $i$-th harmonic number.
We will show that the resulting profile $(\s^*,\t^*)$ remains an equilibrium for the millers, too.
We begin by proving that after inserting the millers in \cref{line:miller-insertion}, the millers are indeed in a miller equilibrium.

\begin{figure}
    \newcommand*{\figalign}{1.465}
    \centering
    \begin{subfigure}[T]{0.21\textwidth}
    \centering
    \begin{tikzpicture}
        \draw[dotted,-] (-0.2,\locsepdist) -- (3.35,\locsepdist);
        
        \node[vert] (x) at (0.9, 0) {$x$};
        \node[vert] (y) at (1.9, 0) {$y$};
        \node[vert] (z) at (2.9, 0) {$z$};

        \foreach \i in {0,1,2,3,4,5} {
            \node[baker] (b\i) at (\i*0.6, \trlength) {};
        }

        \draw[cblue,-] ($(x.center) + (-33.8:5+\noderadius)$) -- ($(b3.center) + (-33.8:7pt)$) arc (-33.8:90:7pt) -- ($(b0.center) + (0, 7pt)$) arc (90:213.8:7pt) -- ($(x.center) + (213.8:5+\noderadius)$) arc (213.8:326.2:5+\noderadius);
        \node[cblue] at (x.center |- 0,-0.65) {$\ell_1$};
        \node[cblue] at (z.center |- 0,-0.65) {$\ell_2$};
        \node[cblue] at (y.center |- 0,-0.65) {$\ell_3$};
        
        \draw[cblue,-] ($(z.center) + (5.2:5+\noderadius)$) -- ($(b5.center) + (5.2:7pt)$) arc (5.2:90:7pt) -- ($(b4.center) + (0, 7pt)$) arc (90:196.8:7pt) -- ($(z.center) + (196.8:5+\noderadius)$) arc (196.8:365.2:5+\noderadius);
        
        \path (b0) edge (x);
        \path (b1) edge (x);
        \path (b2) edge (x);
        \path (b3) edge (x);
        
        \path (b2) edge (y);
        \path (b3) edge (y);
        \path (b4) edge (y);
        
        \path (b3) edge (z);
        \path (b4) edge (z);
        \path (b5) edge (z);

        \useasboundingbox (current bounding box.south west) rectangle (current bounding box.north east |- 0,\figalign);
    \end{tikzpicture}
    \caption{Instance with baker assignments circled in blue.}
    \end{subfigure}
    \begin{subfigure}[T]{0.065\textwidth}
    \begin{tikzpicture}
        \node at (0, 0) {locations};
        \node at (0, \trlength cm) {bakers};
        \draw[dotted,-] (-0.5,\locsepdist) -- (0.5,\locsepdist);

        \useasboundingbox (current bounding box.south west) rectangle (current bounding box.north east |- 0,\figalign);
    \end{tikzpicture}
    \end{subfigure}
    \begin{subfigure}[T]{0.15\textwidth}
    \centering
    \begin{tikzpicture}
        \draw[dotted,-] (-0.3,\locsepdist) -- (2.3,\locsepdist);
        
        \node[vert] (x) at (0, 0) {$x$};
        \node[vert] (z) at (1, 0) {$z$};
        \node[vert] (y) at (2, 0) {$y$};
        \node[cblue] at (x.center |- 0,-0.65) {$\ell_1$};
        \node[cblue] at (z.center |- 0,-0.65) {$\ell_2$};
        \node[cblue] at (y.center |- 0,-0.65) {$\ell_3$};
        \node[baker] at (0, \trlength - 0.375) {};
        \node[baker] at (0, \trlength - 0.125) {};
        \node[baker] at (0, \trlength + 0.125) {};
        \node[baker] at (0, \trlength + 0.375) {};
        \node[baker] at (1, \trlength - 0.375) {};
        \node[baker] at (1, \trlength - 0.125) {};
        \useasboundingbox (current bounding box.south west) rectangle (current bounding box.north east |- 0,\figalign);
    \end{tikzpicture}
    \caption{Result ordered by removal time.}
    \label{fig:loop1b}
    \end{subfigure}
    \caption{Illustration of the first loop of \Cref{alg:ne}: Location~$x$ has the most bakers in its range, so it is chosen first as $\ell_1$ and assigned all possible bakers.
    Of the remaining bakers, location~$z$ has the most in its range.
    Finally, there are no bakers left for location~$y$.
    Observe that the bakers can never move to a location that was removed earlier (i.e. that is to the left in the ordered result in (b)).}
    \label{fig:loop1}
\end{figure}
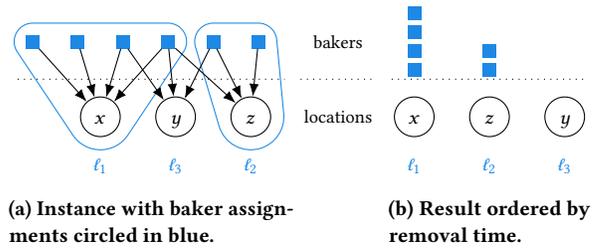

\begin{lemma}[Miller Equilibrium]
\label{lem:miller-eq-greedy-unweighted}
In $(\s,\t^*)$ every miller is allocated to her best-response location.
\end{lemma}
\begin{proof}
Assume for the sake of contradiction there is a miller $m$ on location $\ell$ that could improve to a location $\ell'$, i.e.,
\begin{align}
\frac{B_\s(\ell)}{M_{\t^*}(\ell)} < \frac{B_\s(\ell')}{M_{\t^*}(\ell') + 1}\text.
\end{align}
Consider the last miller that was placed on location $\ell$.
For the number $x$ of millers located at $\ell'$ at this time, we have
\begin{align}
\frac{B_\s(\ell)}{M_{\t^*}(\ell)} \ge  \frac{B_\s(\ell')}{x + 1}
\end{align}
at this time.
Otherwise, $m$ would have been placed on $\ell'$ instead.
As $x \le M_{\t^*}(\ell') $, the lemma follows.
\end{proof}

Before we show that the last step does maintain the property that millers are in equilibrium, let us remark some useful observations.
The number of assigned bakers in $\s$ is non-increasing with the index of the location, i.e., $B_\s(\ell_i) \ge B_\s(\ell_{i+1})$, and each baker $b$ is assigned to the location $\ell_i \in L(b)$ with smallest $i$ of locations in $L(b)$.
This can be seen in \Cref{fig:loop1b}.
Therefore, when moving from $\s$ to $\s^*$, bakers are only moved to locations with higher indices.
Furthermore, by inserting the millers at their best-response locations and breaking ties in favor of smaller indices, we have that $M_{\t^*}(\ell_i) \ge M_{\t^*}(\ell_{i+1})$.

\begin{lemma}[Miller Equilibrium Maintained]
\label{lem:maintain-miller-eq}
Let $\t^*$ be a miller profile computed by \Cref{alg:ne} and let $\s^* \coloneq \arg\max_\x \Phi_{\t^*}(\x)$.
The state $(\s^*,\t^*)$ is a miller equilibrium.
\end{lemma}

\begin{proof}
To prove the lemma, we analyze the changes when moving from $\s$ to $\s^*$.
In particular, we show that the minimal utility of any miller does not decrease and that the maximal achievable utility of a deviation does not increase.
To that end, we denote by $L_{\t^*} \subseteq \L$ the subset of locations occupied by millers in $\t^*$.
We consider the \emph{difference graph} which is a multi-graph $G \coloneq (L_{\t^*},E)$ that describes the change of bakers from $\s$ to $\s^*$.
See \Cref{fig:difference-graph} for an example.
Therefore, $E$ contains an arc from location $\ell_i$ to location $\ell_j$ for each baker $b$ with $s_b = \ell_i$ and $s^*_b = \ell_j$.
Note that we always have $i<j$ and, therefore, the difference graph $G$ is acyclic.

\begin{figure}[h]
    \centering
    \begin{subfigure}[m]{0.29\textwidth}
    \centering
    \begin{tikzpicture}
        \node[anchor=west] at (-1.5, 0) {locations};
        \node[anchor=west] at (-1.5, \trlength cm) {bakers};
        \node[anchor=west] at (-1.5, -\trlength cm) {millers};
        
        \draw[dotted,-] (-1.4,\locsepdist) -- (3.5,\locsepdist);
        \draw[dotted,-] (-1.4,-\locsepdist) -- (3.5,-\locsepdist);
        \foreach \name / \i / \val in {v/0/11,w/1/7,x/2/6,y/3/4,z/4/4} {
            \node[vert] (\name) at (0.8*\i, 0) {$\name$};
            \node[baker, minimum height=\val*3pt, anchor=south, label={[text height=0.5ex]above:$\val$}] (b\name) at (0.8*\i, 0.7) {};
        }
        \node[miller] at (0, -\trlength + 0.15) {};
        \node[miller] at (0, -\trlength - 0.15) {};
        \node[miller] at (0.8*1, -\trlength) {};
        \node[miller] at (0.8*2, -\trlength) {};
        \node[miller] at (0.8*3, -\trlength) {};
        \node[miller] at (0.8*4, -\trlength) {};

        \path (bw) edge (bx);
        \path (bx) edge (by);
        \path[bend left = 15] (bv.north east) edge (bz.north west);
    \end{tikzpicture}
    \end{subfigure}
    \begin{subfigure}[m]{0.18\textwidth}
    \centering
    \begin{tikzpicture}
        \node at (1,2) {};
        \node[vert, label=left:$P_1$] (v) at (0, 1) {$v$};
        \node[vert, label=left:$P_2$] (w) at (0, 0) {$w$};
        \node[vert] (x) at (0.9*1, 0) {$x$};
        \node[vert] (y) at (0.9*2, 0) {$y$};
        \node[vert] (z) at (0.9*2, 1) {$z$};
        
        \path (w) edge (x) (x) edge (y);
        \path (v) edge (z);
    \end{tikzpicture}
    \end{subfigure}
    \caption{For \Cref{lem:maintain-miller-eq} and the state $(\s,\t^*)$ on the left, the arrows indicate the strategy changes of bakers from $\s$ to $\s^*$, e.g., there is a baker $b$ with $s_b = v$ and $s_b^* = z$.
    This results in the difference graph $G$ on the right, with the two paths $P_1$ and $P_2$ making up the path decomposition $\mathcal{P}$.}
    \label{fig:difference-graph}
\end{figure}
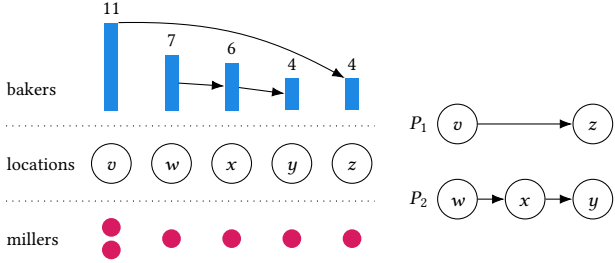

We consider a path decomposition $\mathcal{P}$ of $G$ into maximal paths.
In particular, there is no location that is both a start and an endpoint of paths, as these two paths would be merged.
Let $(p_1,p_2,\dots,p_m) \in \mathcal{P}$ be a path consisting of locations.
As $\s^*$ maximizes the function $\Phi_{\t^*}(x) \coloneq \sum_{\ell \in \L} M_{\t^*}(\ell) \cdot H_{B_x(\ell)}$, we have that
\begin{align}
    \frac{M_{\t^*}(p_1)}{B_{\s^*}(p_1) +1}\le \frac{M_{\t^*}(p_m)}{B_{\s^*}(p_m)} \ \ \text{or equiv.}\ \  \frac{B_{\s^*}(p_1) +1}{M_{\t^*}(p_1)}\ge \frac{B_{\s^*}(p_m)} {M_{\t^*}(p_m)}\label{potential-improvement}
\end{align}
since otherwise rearranging the bakers backwards along that path would increase the value of $\Phi$.
Note that rearranging changes only the number of bakers on $p_1$ and $p_m$.

By the choice of paths in $\mathcal{P}$ being \emph{maximal}, we have that $p_1$ has only outgoing arcs and $p_m$ has only incoming arcs in $G$.
Therefore,
\begin{align}
    B_{\s^*}(p_1) +1 \le B_{\s}(p_1) \quad\text{and}\quad B_{\s^*}(p_m) \ge B_{\s}(p_m)+1\text.\label{monotone-improvement}
\end{align}
Combining \Cref{potential-improvement,monotone-improvement} we get that the maximal utility of any miller for baker profile $\s^*$ is not larger than for $\s$:
\begin{align}
    \frac{B_{\s}(p_1)}{M_{\t^*}(p_1)} \ge \frac{B_{\s^*}(p_1)+1}{M_{\t^*}(p_1)}  \ge  \frac{B_{\s^*}(p_m)}{M_{\t^*}(p_m)} \ge \frac{B_{\s}(p_m)+1}{M_{\t^*}(p_m)}\text.
    \label{improvement-maximal}
\end{align}
Since we have $M_{\t^*}(p_1) \ge M_{\t^*}(p_m)$ because $p_1$ has a lower index in $(\ell_1,\dots, \ell_{|\L|})$ than $p_m$, we also have
\begin{align}
    \frac{B_{\s*}(p_1)}{M_{\t^*}(p_1)}  \ge  \frac{B_{\s}(p_m)}{M_{\t^*}(p_m)}
    \label{improvement-minimal}
\end{align}
which implies that the minimal utility of any miller for $\s^*$ is not smaller than for $\s$.
Now assume for the sake of contradiction that there is a miller $m$ that has a profitable deviation from location $\ell \coloneq \t^*(m)$ to location $\ell'$ in $(\s^*,\t^*)$. 
Furthermore, let $m'$ be the miller with minimal utility $u_\text{min}$ in $(\s^*,\t)$.
Therefore, we have that
\begin{align}
    \frac{B_{\s^*}(\ell')}{M_{\t^*}(\ell')+1}  \ge  \frac{B_{\s^*}(\ell)}{M_{\t^*}(\ell)}  \ge   u_\text{min}
    \label{deviation}
\end{align}
where the second inequality follows from \Cref{improvement-minimal}.
However, by ~\Cref{lem:miller-eq-greedy-unweighted} this deviation for $m'$ was not profitable in $(\s,\t^*)$:
\begin{align}
    \frac{B_{\s}(\ell')}{M_{\t^*}(\ell')+1}  \le  u_\text{min}\text.
    \label{no-min-improve}
\end{align}
Combining \Cref{deviation,no-min-improve} implies that
$B_{\s^*}(\ell') > B_{\s}(\ell')$, so $\ell'$ is the endpoint of some path in $\mathcal{P}$.
We now consider a path in the difference graph $G$ that ends at $\ell'$ and starts from some location $\ell_\text{start}$.
By \Cref{improvement-maximal} we have that
\begin{align}
    \frac{B_{\s}(\ell_\text{start})}{M_{\t^*}(\ell_\text{start})}  \ge  \frac{B_{\s^*}(\ell')}{M_{\t^*}(\ell')}\text.
\end{align}
By the existence of a path from $\ell'$ to $\ell_\text{start}$, we get $M_{\t^*}(\ell_\text{start}) \ge M_{\t^*}(\ell')$.
Therefore, we have
\begin{align}
    \frac{B_{\s}(\ell_\text{start})}{M_{\t^*}(\ell_\text{start}) +1}  \ge  \frac{B_{\s^*}(\ell')}{M_{\t^*}(\ell')+1}\text.
\end{align}
However, then deviating to $\ell_\text{start}$ would already have been an improving deviation for $m'$ in $(\s,\t^*)$, which contradicts \Cref{lem:miller-eq-greedy-unweighted}.
Therefore, if no miller has a profitable unilateral deviation for baker profile $\s$, there is also none in $\s^*$, which concludes the lemma.
\end{proof}

Finally, it remains to show that $(\s^*,\t^*)$ is a baker equilibrium, which follows from the fact that $\Phi$ is an exact potential function for the induced baker game if we fix all millers.
Assume for the sake of contradiction that there is a baker $b$ that has a profitable deviation from $\ell = s^*_b$ to $\ell'$, with
\begin{align}
\frac{M_{\t^*}(\ell)}{B_{\s^*}(\ell)} < \frac{M_{\t^*}(\ell')}{B_{\s^*}(\ell') +1}\text.
\end{align}
But then
\begin{align*}
    \Phi_{\t^*}((\s^*_{-b},\ell'))=\Phi_{\t^*}(\s^*) - \frac{M_{\t^*}(\ell)}{B_{\s^*}(\ell)} + \frac{M_{\t^*}(\ell')}{B_{\s^*}(\ell') +1} > \Phi_{\t^*}(\s^*)\text,
\end{align*}
which contradicts the choice of $\s^*$ as a maximizer of $\Phi_{\t^*}$.
In the following \Cref{lem:compute-maximizer}, we show how to compute the maximizer in polynomial time and thus our current \Cref{thm:ne} follows.
\end{proof}

To compute the maximizer in \cref{line:maximizer} of \Cref{alg:ne}, we employ a reduction to integral \textsc{MinCostFlow} similar to the algorithm of \citet*[Theorem 2]{potential-min-algo} for symmetric congestion games.

\begin{lemma}[Potential Function Maximizer]
    \label{lem:compute-maximizer}
    The maximizer of $\Phi_{\t^*}(s) \coloneq \sum_{\ell \in \L} M_{\t^*}(\ell) \cdot H_{B_s(\ell)}$ can be computed in polynomial time.
\end{lemma}
\begin{proof}
The set of nodes of the \textsc{MinCostFlow} instance corresponds to the union of the set of bakers $\B$, the set of locations $\L$, and the set of a source $s$ and a sink $t$.
There are edges with capacity $1$ and cost $0$ from $s$ to all bakers and from each baker to her accessible locations.
From each location $\ell$, there are $|\B|$ parallel edges to $t$ with capacity $1$ and costs $ - \frac{M_{\t^*}(\ell)}{1}, - \frac{M_{\t^*}(\ell)}{2}, \dots, - \frac{M_{\t^*}(\ell)}{|\B|}$.
We give an example of this construction in \Cref{fig:mincostflow}.

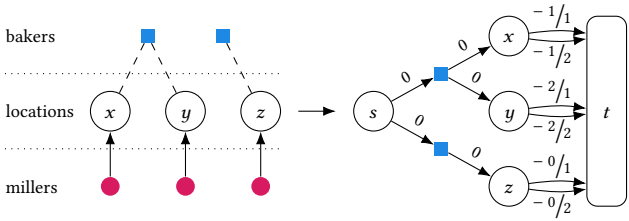
\begin{figure}
    \centering
    \begin{tikzpicture}
        \node[anchor=west] at (-1.5, 0) {locations};
        \node[anchor=west] at (-1.5, -\trlength) {millers};
        \node[anchor=west] at (-1.5, \trlength) {bakers};
        
        \draw[dotted,-] (-1.4,\locsepdist) -- (2.3,\locsepdist);
        \draw[dotted,-] (-1.4,-\locsepdist) -- (2.3,-\locsepdist);
        
        \node[vert] (s) at (3.5, 0) {$s$};
        \node[vert, rectangle, rounded corners, minimum height = 2*\noderadius + 2 cm] (t1) at (6.6,0) {$t$};
        \node[vert, rectangle, draw=none] (t0) at (6.6, 1) {};
        \node[vert, rectangle, draw=none] (t2) at (6.6,-1) {};
        
        \path (2.5,0) edge (3,0);
        
        \foreach \i in {0, 1} {
            \node[baker] (norm_b\i) at (0.5+1*\i, \trlength) {};
            \node[baker] (reduc_b\i) at (4.4, 0.5-1*\i) {};
            \path (s) edge["0"] (reduc_b\i);
        }
        \foreach \i / \name in {0/x, 1/y, 2/z} {
            \node[vert] (norm_\name) at (1*\i, 0) {$\name$};
            \node[miller] (norm_m\i) at (1*\i, -\trlength) {};
            \node[vert] (reduc_\name) at (5.3, 1-1*\i) {$\name$};
        }
        \foreach \name / \options / \cost in {reduc//0, norm/{unchosen}/} {
            \path[\options] (\name_b0) edge (\name_x);
            \path[\options] (\name_b0) edge["\cost"] (\name_y);
            \path[\options] (\name_b1) edge["\cost"] (\name_z);
        }
        \foreach \i / \name / \millers in {0/x/1, 1/y/2, 2/z/0} {
            \path (reduc_\name) edge[bend left=8, "{\large$\sfrac{-\millers}{1}$}" {xshift=-3pt,yshift=-2pt}] (t\i);
            \path (reduc_\name) edge[bend right=8, "{\large$\sfrac{-\millers}{2}$}" {xshift=-3pt,yshift=2pt, below}] (t\i);
        }
        \path (norm_m0) edge (norm_x);
        \path (norm_m1) edge (norm_y);
        \path (norm_m2) edge (norm_z);
    \end{tikzpicture}
    \caption{To compute the maximizer of $\Phi_{\t^*}$, the instance on the left with the given assignment of millers $\t^*$ is reduced to the \textsc{MinCostFlow} instance on the right by \Cref{lem:compute-maximizer}. All edge capacities are $1$ and the costs are given by the labels.}
    \label{fig:mincostflow}
\end{figure}

Since in any integral \textsc{MinCostFlow} for a given location only the edges with minimal cost are used, the cost of a flow is equal to $-\Phi_{\t^*}(\s) = -\sum_{\ell \in \L} M_{\t^*}(\ell) \cdot H_{B_\s(\ell)}$ for the corresponding baker assignment $\s$.
Thus, an integral \textsc{MinCostFlow} in this network corresponds to an assignment of bakers to locations that maximizes the function $\Phi_{\t^*}$.

Note that our instance contains negative costs but no negative cycles. Therefore, we may use Orlin's algorithm~\cite{orlin}.
\end{proof}

We observe that the runtime of \Cref{alg:ne} is dominated by the computation of the \textsc{MinCostFlow}. The construction yields an instance with $O(|\B|+|\L|)$ nodes and $O(|\B||\L|)$ edges.
In \Cref{sec:approx}, we will give a bound on the social welfare approximation of the algorithm.

\section{Social Welfare}
\label{sec:welfare}

The most prominent notion of social welfare is utilitarian social welfare, i.e., the summation of the agents' utilities.
However, in the game at hand, this sum essentially simplifies to the number of millers plus the number of covered bakers, where the latter denotes bakers for which there is a miller on a location they can choose:
\begin{align}
    \text{Welfare} =& \sum_{b\in \B} u_b(\s,\t) + \sum_{m\in \M} u_m(\s,\t)\notag \\
    =& \sum_{\ell : B_\s(\ell) \ne 0} B_\s(\ell) \frac{M_\t(\ell)}{B_\s(\ell)} + \sum_{\ell : M_\t(\ell) \ne 0} M_\t(\ell) \frac{B_\s(\ell)}{M_\t(\ell)}\notag \\
    =& \sum_{\ell : B_\s(\ell) \ne 0} M_\t(\ell) + \sum_{\ell : M_\t(\ell) \ne 0} B_\s(\ell)\text. \label{l9}
\end{align}
Note that in equilibrium, the first term of \Cref{l9} evaluates to $|\M|$ and the second term is equal to the number of covered bakers.
We therefore define the latter as \emph{coverage} and use it as a metric for social welfare in the rest of the paper.
\begin{definition}[Coverage]
The coverage of a strategy profile $(\s,\t)$ is the number of bakers with at least one miller at the same location:
\[
    W(\s, \t) \coloneq |\{b \in \B \mid M_\t(s_b) > 0\}|\text.
\]
\end{definition}
For this measure, we show that the price of anarchy is linear in $|\B|$, while the price of stability is $1 + \frac{\min(|\L|,|\M|)-1}{|\M|}$.
We also show that computing both the social welfare optimum and the optimal Nash equilibrium is NP-hard.
On the positive side, we show that our \Cref{alg:ne} actually computes a Nash equilibrium that approximates the optimal social welfare by a factor of $\left(1 + \frac{\min(|\L|,|\M|)-1}{|\M|}\right)\frac{e}{e-1}$.

\subsection{Price of Anarchy and Stability}

We show that the price of anarchy is $|\B|$, which is the worst possible in terms of the number of bakers.
This follows from the fact that in every instance there exists a Nash equilibrium that covers at least one baker, while for some instances there is an equilibrium that covers exactly one baker.
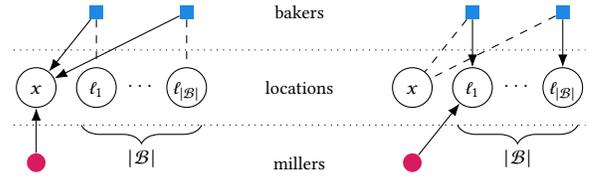
\begin{figure}[h]
    \centering
    \begin{tikzpicture}
        \node at (3.5, 0) {locations};
        \node at (3.5, \trlength) {bakers};
        \node at (3.5, -\trlength) {millers};
        \draw[dotted,-] (-0.3,\locsepdist) -- (5+2.3,\locsepdist);
        \draw[dotted,-] (-0.3,-\locsepdist) -- (5+2.3,-\locsepdist);

        \foreach \var / \offset \ in {a/0,b/5} {
            \node[vert] (l0\var) at (\offset+ 0, 0) {$x$};
            \node[vert] (l2\var) at (\offset+ 0.8, 0) {$\ell_1$};
            \node at (\offset+ 1.4, 0) {$\,\cdots$};
            \node[vert] (l5\var) at (\offset+ 2, 0) {$\ell_{|\B|}$};
            
            \node[miller] (f\var) at (\offset, -\trlength) {};
            
            \foreach \i in {2,5} {
                \node[baker] (c\i\var) at (\offset+ 0.4*\i, \trlength) {};
            }
            \draw[decorate,decoration={brace,amplitude=7pt},-] (\offset + 2.2, -\locsepdist) -- (\offset +0.6, -\locsepdist) node [midway,yshift=-0.45cm] {$|\B|$};
        }
        \path (fa) edge (l0a);
        \path (fb) edge (l2b);
        \path (c2a) edge (l0a);
        \path (c5a) edge (l0a);
        \path (c2b) edge (l2b);
        \path (c5b) edge (l5b);
        
        \path[unchosen] (c2b) edge (l0b);
        \path[unchosen] (c5b) edge (l0b);
        \path[unchosen] (c2a) edge (l2a);
        \path[unchosen] (c5a) edge (l5a);
    \end{tikzpicture}
    \caption{This instance illustrates the worst case for the price of anarchy.
    The bakers are restricted to $x$ and the location directly below them while the miller is unrestricted.
    Left: The social welfare optimum, right: a pure Nash equilibrium.}
    \label{fig:poa}
\end{figure}

\begin{theorem}[Price of Anarchy]
    The price of anarchy is $|\B|$.
\end{theorem}

\begin{proof}

The price of anarchy is trivially at most $|\B|$ since at least one of the bakers is covered by a miller in equilibrium.
For a matching lower bound, we construct the following instance:
Let the set of locations be $\L \coloneq \{x, \ell_1, \dots, \ell_{|\B|}\}$ and let there be one miller.
For each baker $b_i \in \{b_1, \dots, b_{|\B|}\}$, let the set of permissible locations be $L(b_i) \coloneq \{\ell_i, x\}$.
There is a pure Nash equilibrium where each baker $b_i$ is placed on location~$\ell_i$ and the miller chooses $\ell_1$ with only one baker covered.
The social optimum has all agents on location~$x$.
See \Cref{fig:poa}.
Thus, we have $\text{PoA}\geq|\B|$.
\end{proof}

In contrast to the low welfare of the worst equilibrium, the best equilibrium is surprisingly good.
\begin{theorem}[Price of Stability]
    The PoS is $1 + \frac{\min(|\L|,|\M|)-1}{|\M|}$.
\end{theorem}

The theorem follows immediately from the following two lemmas.
For the upper bound, we use the strategy profile with the highest coverage and construct a pure Nash equilibrium where the millers use a subset of the miller location of this optimal state.

\begin{lemma}[Price of Stability Upper Bound]
    \label{lem:pos-upper}
    There exists a pure Nash equilibrium that approximates the optimal coverage by a factor of $1 + \frac{\min(|\L|,|\M|)-1}{|\M|}$.
\end{lemma}
\begin{proof}
    Let $L_\text{opt}$ denote the $q \coloneq \min(|\L|,|\M|)$ locations that maximize the number of covered bakers, i.e., bakers~$b$ with $L(b) \cap L_\text{opt} \neq \varnothing$.
    Clearly, the state \text{OPT} where one miller is placed on each of the locations in $L_\text{opt}$, while all bakers are assigned to a location in $L_\text{opt}$ if possible, has the optimal coverage.

    To obtain an equilibrium state $\text{NE}$ with the claimed coverage, we run \Cref{alg:ne} on a modified instance in which we remove all locations that are not in $L_\text{opt}$.
    To complete $\text{NE}$, we add the previously removed locations and assign bakers that have no available location in $L_\text{opt}$ to arbitrary permissible locations.

    We show that $\text{NE}$ is still a pure Nash equilibrium:
    Bakers outside $L_\text{opt}$ have no possible location available to them inside $L_\text{opt}$.
    By \Cref{thm:ne}, bakers inside $L_\text{opt}$ have no profitable deviating move to another location inside $L_\text{opt}$. Furthermore,
    Bakers cannot profit from moving outside $L_\text{opt}$, because there are no millers there.
    By \Cref{thm:ne}, also the millers have no profitable deviating move to a location inside $L_\text{opt}$.
    It remains to show that millers do not want to move outside $L_\text{opt}$:
    Note that in NE, every baker that can be assigned to a location in $L_\text{opt}$ is assigned to a location in $L_\text{opt}$.
    Thus, each of the $|\M|$ locations inside $L_\text{opt}$ has at least as many bakers assigned to it as the most profitable location $\ell$ outside $L_\text{opt}$.
    Otherwise, we could construct a set $L_\text{opt}'$ which covers more bakers than $L_\text{opt}$ by swapping in $\ell$.
    If not $L_\text{opt}=\L$, then $L_\text{opt}$ has $|\M|$ locations and occupying a location in $L_\text{opt}$ as the only miller is possible. Therefore, in $\text{NE}$ any miller's utility is at least as large as the number of bakers on any location outside $L_\text{opt}$.

    It remains to bound the social welfare.
    We denote by $n$ the minimal utility that a miller receives in the equilibrium $\text{NE}$ which implies a coverage of at least $|M|n$.
    On the other hand, consider any location $\ell \in L_\text{opt}$ that is not covered by any miller.
    Equilibrium conditions imply that there are at most $n$ bakers on $\ell$ as otherwise, the miller with minimal utility could improve.
    There are at most $q-1$ uncovered locations in $L_\text{opt}$ for NE, so there are at most $(q-1)n$ bakers covered in OPT but not in NE.
    Thus, we have
\begin{align*}
    \frac{W(\text{OPT})}{W(\text{NE})} \leq \frac{W(\text{NE}) + (q-1)n}{W(\text{NE})} = 1+ \frac{(q-1)n}{W(\text{NE})}& \\
    \leq 1+ \frac{(q-1)n}{|\M|n} = 1 + \frac{\min(|\L|,|\M|)-1}{|\M|}\text.&\qedhere
\end{align*}
\end{proof}

Note that, the algorithm implied by this lemma is more complex than \Cref{alg:ne} because it relies on knowing the set $L_\text{opt}$ with optimal coverage of bakers.
Next, we match the upper bound with a lower bound on the quality of the best pure Nash equilibrium:

\begin{lemma}[Price of Stability Lower Bound]
    The price of stability is at least $1 + \frac{\min(|\L|,|\M|)-1}{|\M|}$.
\end{lemma}

\begin{proof}
We construct the following instance for a given integer $n$:
Let the set of locations be $\L \coloneq \{x, \ell_2, \dots, \ell_{|\L|}\}$.
For location $x$ we have $n|\M|+1$ bakers restricted to only $x$.
For each remaining location $\ell$, we have $n$ bakers restricted to only $\ell$.
The Nash equilibrium has all millers on $x$.
The social optimum covers $q \coloneq \min(|\L|,|\M|)$ locations including $x$.
See \Cref{fig:pos}.
So we have
\[
    \text{PoS} \geq \frac{|\M|n+1 + n(q-1)}{|\M|n+1}= 1 + \frac{n(q-1)}{|\M|n+1}
\]
\[
    \text{with}\quad \lim_{n\to\infty}\left(1 + \frac{n(q-1)}{|\M|n+1}\right) = 1 + \frac{\min(|\L|,|\M|)-1}{|\M|}\text.\qedhere
\]
\end{proof}

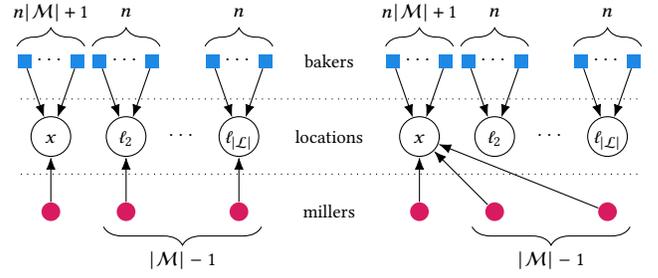
\begin{figure}[h]
    \centering
    \begin{tikzpicture}
        \node at (3.7, 0) {locations};
        \node at (3.7, \trlength) {bakers};
        \node at (3.7, -\trlength) {millers};
        \draw[dotted,-] (-0.4,\locsepdist) -- (4.9+2.9,\locsepdist);
        \draw[dotted,-] (-0.4,-\locsepdist) -- (4.9+2.9,-\locsepdist);
        
        \foreach \var / \offset / \labelloc \ in {a/0/below right,b/4.9/below} {
            \node[vert] (l0\var) at (\offset+ 0, 0) {$x$};
            \node[vert] (l2\var) at (\offset+ 1, 0) {$\ell_2$};
            \node at (\offset+ 1.75, 0) {$\,\cdots$};
            \node[vert] (l5\var) at (\offset+ 2.5, 0) {$\ell_{|\L|}$};
            
            \foreach \i / \desc in {0/$n|\M|+1$,2/$n$,5/$n$} {
                \node[miller] (f\i\var) at (\offset+ 0.5*\i, -\trlength) {};
                
                \node[baker] (c0\i\var) at (\offset-0.35+ 0.5*\i, \trlength) {};
                \node at (\offset+ 0.5*\i, \trlength) {$\,\cdots$};
                \node[baker] (c1\i\var) at (\offset+0.35+ 0.5*\i, \trlength) {};
                \path (c0\i\var) edge (l\i\var);
                \path (c1\i\var) edge (l\i\var);
                
                \draw[decorate,decoration={brace,amplitude=7pt},-] (\offset-0.45+ 0.5*\i, 1.2*\trlength) -- (\offset+0.45+ 0.5*\i, 1.2*\trlength) node [midway,yshift=0.45cm] {\desc};
            }
            \draw[decorate,decoration={brace,amplitude=7pt},-] (\offset + 2.8, -1.2*\trlength) -- (\offset +0.7, -1.2*\trlength) node [midway,yshift=-0.45cm] {$|\M|-1$};
        }
        \foreach \i in {0,2,5} {
            \path (f\i a) edge (l\i a);
            \path (f\i b) edge (l0b);
        }
    \end{tikzpicture}
    \caption{An instance with $|\L|=|\M|$ with the price of stability approaching $1 + \frac{\min(|\L|,|\M|)-1}{|\M|}$ as $n \rightarrow \infty$.
    Left: The social welfare optimum, right: the only pure Nash equilibrium.
    \label{fig:pos}}
\end{figure}

\subsection{Complexity and Approximation of Optimal Solutions}
\label{sec:approx}
We first show hardness of computing optimal solutions and optimal equilibria.
To that end, we reduce from the \textsc{Maximum $k$-Coverage} problem~\cite{maxcov}.

\begin{definition}[\textsc{Maximum $k$-Coverage}]
    Given an integer $k$ and a set of sets $S$, find a subset $S' \subseteq S$ with $|S'| = k$ such that the number of included items $\left|\bigcup_{T \in S'}{T}\right|$ is maximal.
\end{definition}

We now use the NP-hardness of \textsc{Maximum $k$-Coverage} to show that finding the social welfare optimum and the socially optimal Nash equilibrium is NP-hard.

\begin{theorem}[Social Welfare Optimum]
    Computing the social welfare optimum is NP-hard.
\end{theorem}
\begin{proof}
    We reduce from \textsc{Maximum $k$-Coverage}.
    Let the set of locations be $\L \coloneq S$ and for each item $s$ in the base set $\bigcup_{T \in S}{T}$, we create a baker $b$ with $L(b) \coloneq \{T \in S \mid s \in T\}$.
    We add $k$ millers.

    If there is an optimal solution $S'$ that covers $n$ items, then for the corresponding strategy profile $(\s, \t)$ we place the millers at exactly the locations $S'$ and let all bakers choose a location in $S'$ if possible.
    The result is a coverage of $W(\s, \t)=n$.

    For a social optimum $(\s, \t)$, we set $S'$ to the set of locations that have at least one miller placed on them.
    Since $(\s, \t)$ is a social optimum, all bakers that can be placed in $S'$ are placed in $S'$ and therefore the number of covered bakers is equal to the number of items contained in $\bigcup_{T \in S'}{T}$.
\end{proof}

For optimal pure Nash equilibria, the reduction is more involved:
Adding dummy bakers to each location ensures there is no pure Nash equilibrium with multiple millers on the same location.

\begin{theorem}[Socially Optimal Equilibrium]
    Computing the pure Nash equilibrium with the best coverage is NP-hard.
\end{theorem}
\begin{proof}
    Again, we reduce from \textsc{Maximum $k$-Coverage}.
    Let the set of locations be $\L \coloneq S$ and for each item $s$ in the base set $\bigcup_{T \in S}{T}$, we create a baker $b$ with $L(b) \coloneq \{T \in S \mid s \in T\}$.
    For each location $\ell$, we also add $q = \left|\bigcup_{T \in S}{T}\right| + 1$ bakers $b$ with $L(b)=\{\ell\}$.
    We add $k$ millers.
    We show that there is an optimal solution $S'$ that covers $n$ items if and only if there is a socially optimal pure Nash equilibrium $(\s, \t)$ with coverage $W(\s, \t)=n+kq$.

    If there is an optimal solution $S'$ that covers $n$ items, then for the corresponding miller profile $\t$ we place the millers at exactly the locations $S'$.
    Given $\t$, we choose an arbitrary baker equilibrium $\s$.
    Note that in $\s$ every baker that has a location available in $S'$ is placed in $S'$.
    The state $(\s, \t)$ is a pure Nash equilibrium because a miller that profitably deviates to a location $\ell\notin S'$ implies that there is a set $S''$ with more covered items.
    A miller cannot deviate to a location inside $S'$ because, with two millers at the same location, both receive a utility of less than $q$.
    Thus, we have a pure Nash equilibrium with a coverage of $W(\s, \t)=n+kq$.

    For a socially optimal Nash equilibrium $(\s, \t)$, we set $S'$ to the set of locations that have at least one miller placed on them.
    In $(\s, \t)$ no two millers are placed at the same location since that would yield a utility below $q$.
    Since $(\s, \t)$ is a socially optimal Nash equilibrium, all bakers that have an available location in $S'$ choose a location in $S'$.
    Therefore the number of covered bakers is equal to $n$, the number of items contained in $\bigcup_{T \in S'}{T}$ with $kq$ extra bakers added.
\end{proof}

With a similar technique as in our proof for the upper bound of the price of stability in \Cref{lem:pos-upper}, we show that the pure Nash equilibrium computed by our polynomial-time \Cref{alg:ne} is a good approximation of the best Nash equilibrium in terms of social welfare.
The added factor $\frac{e}{e-1}$ with Euler's number $e$ results from the fact, that our algorithm greedily selects miller locations instead of choosing them optimally.

\begin{theorem}[Nash Equilibrium Welfare Approximation]
    \Cref{alg:ne} computes a pure Nash equilibrium that approximates the optimal coverage by a factor of $\left(1 + \frac{\min(|\L|,|\M|)-1}{|\M|}\right)\frac{e}{e-1}$.
\end{theorem}
\begin{proof}
    We note that the greedy algorithm for the \textsc{Maximum $k$-Coverage} problem has an approximation ratio of $\frac{e}{e-1}$~\cite{maxcov-greedy}.
    We observe that for $q \coloneq \min(|\L|,|\M|)$, this greedy solution picks exactly the locations $\ell_1,\ldots,\ell_{q}$ as defined by \Cref{alg:ne} to place the millers.
    Furthermore, the algorithm assigns millers to a subset of those locations, i.e., $\ell_1,\ldots,\ell_x$ for some $x \leq q$ because the algorithm iteratively assigns millers to their best response location.
    Hence, it either assigns millers to an already covered location or the next uncovered location since the number of bakers on location $\ell_i$ is non-increasing with $i$.

    Again, we denote by $n$ the minimal utility that a miller receives in the equilibrium computed by \Cref{alg:ne}.
    Then the coverage in the equilibrium is at least $|\M|n$.
    On the other hand, consider any location $\ell \in \ell_{x+1},\ldots,\ell_{q}$ that is not covered by a miller.
    By the sorting of the locations, there are at most $n$ bakers on $\ell$.
    As there are at most $q-1$ uncovered locations in $\ell_1,\ldots,\ell_{q}$, the total number of bakers covered by the greedy solution but not \Cref{alg:ne} is at most $(q-1)n$.
    Thus, the ratio between the coverage of the greedy solution and the equilibrium computed by \Cref{alg:ne} is at most
    \[
        \frac{|\M|n + (q-1)n}{|\M|n} = 1 + \frac{\min(|\L|,|\M|)-1}{|\M|}\text.
    \]
    Combining this with the factor $\frac{e}{e-1}$ of the greedy approximation yields the theorem.
\end{proof}

\section{Conclusion and Future Work}
Our model adds location choice to the classical setting of simple symmetric Fractional Hedonic Games and to Hedonic Diversity Games with single-peaked utilities.
This extension renders proving the existence of pure Nash equilibria much more challenging and we cope with this problem by carefully designing a three-stage constructive procedure that ensures that neither bakers nor millers want to deviate from the constructed state of the game. It is still an open question if this result can also be proven via a potential function approach.

The computed pure Nash equilibrium approximates the best possible social welfare by a factor of $\left(1 + \frac{\min(|\L|,|\M|)-1}{|\M|}\right)\frac{e}{e-1}$.
Note that the individual factors are tight:
The price of stability is $1 + \frac{\min(|\L|,|\M|)-1}{|\M|}$ and $\frac{e}{e-1}$ is the best achievable factor for selecting locations unless $\text{P}=\text{NP}$~\cite{setcover-inapprox}.
However, it is an open problem, if the combination of both factors is also a tight upper bound.

An interesting future direction is to add weights to the agents.
Our results on the price of anarchy and stability also hold in this setting with only minor modifications of the proofs.
The bakers and the millers individually still play Congestion Games, but already for weighted agents on only one side, there is no potential function as can be seen by the improving response cycle in~\Cref{fig:cycle}.
\begin{figure}[h]
    \centering
    \begin{subfigure}{0.566\columnwidth}
    \centering
    \begin{tikzpicture}
        \node[anchor=west] at (-2.8, 0) {locations};
        \node[anchor=west] at (-2.8, \trlength) {bakers};
        \node[anchor=west] at (-2.8, -\trlength) {millers};
        \node[vert] (x) at (-1, 0) {$x$};
        \node[vert] (y) at (0, 0) {$y$};
        \node[vert] (z) at (1, 0) {$z$};
        
        \draw[dotted,-] (-2.7,-\locsepdist) -- (1.75,-\locsepdist);
        \draw[dotted,-] (-2.7,\locsepdist) -- (1.75,\locsepdist);

        \foreach \i in {1,2,3,4,5,6,7,8,9,10,11,12} {
            \pgfmathsetmacro{\facloc}{(\i - 6.5) * 0.3}
            \node[miller] (f\i) [label=below:$1$] at (\facloc, -\trlength) {};
        }

        \node[baker] (c1) [label=above:$5$] at (-1.33, \trlength) {};
        \node[baker] (c2) [label=above:$8$] at (-0.67, \trlength) {};
        \node[baker] (c3) [label=above:$8$] at (0, \trlength) {};
        \node[baker] (c4) [label=above:$5$] at (0.67, \trlength) {};
        \node[baker] (c5) [label=above:$6$] at (1.33, \trlength) {};

        \foreach \i in {1,2,3,4,5,6} {
            \path (f\i) edge (x);
        }
        \foreach \i in {7,8} {
            \path (f\i) edge (y);
        }
        \foreach \i in {9,10,11,12} {
            \path (f\i) edge (z);
        }
        \path (c1) edge (x);
        \path (c2) edge (x);
        \path (c3) edge (y);
        \path (c4) edge (z);
        \path (c5) edge (z);
    \end{tikzpicture}
    \vspace{6.5pt}
    \caption{Instance and starting configuration.}
    \end{subfigure}
    \begin{subfigure}{0.425\columnwidth}
    \centering
    \setlength{\tabcolsep}{4.5pt}
    \begin{tabular}{lcr}
        \toprule
        Type&Move&Weight\\
        \midrule
        Miller&$x\to z$&1\\
        Baker&$y\to z$&8\\
        Baker&$x\to y$&5\\
        Miller&$x\to y$&1\\
        Baker&$z\to y$&6\\
        Miller&$x\to z$&1\\
        Miller&$x\to y$&1\\
        \bottomrule
    \end{tabular}
        
        
        
        
        
        
        
        
    \caption{Moves in the cycle.}
    \end{subfigure}
    \caption{An instance with weighted bakers and unweighted millers that admits an improving response cycle.
    The edges denote the starting state locations for the cycle.
    All agents may choose all locations.}
    \label{fig:cycle}
\end{figure}

Another possible research direction is to restrict the millers in their location choice as well. Here, we cannot rule out the existence of a potential function. However, settling this open problem seems to be challenging, since it also remains unsolved for the related model by~\citet{ijcai-23}.

Finally, as already mentioned above, other variants of Hedonic Games can also be extended by adding a layer of location choice. We believe that investigating the impact of restricting the set of possible coalitions in such a natural way will lead to interesting new insights into coalition formation problems. Using locations, effects like geographical proximity, access restrictions, and space restrictions can be combined in simple and clean models.

\begin{acks}
This work was supported by the Federal Ministry of Education and Research (BMBF) with a fellowship within the IFI program of the German Academic Exchange Service (DAAD).
\end{acks}

\balance
\bibliographystyle{ACM-Reference-Format}
\bibliography{simultaneous}
\end{document}